\documentclass[11pt,reqno]{amsart}
\usepackage{amssymb,mathrsfs,color}
\usepackage{pinlabel}

\usepackage{amsmath, amsfonts, amsthm, verbatim}
\usepackage{epstopdf, mathtools}
\mathtoolsset{showonlyrefs}
\usepackage{color}
\usepackage{esint}
\usepackage{stmaryrd}

\usepackage{empheq}
\usepackage[shortlabels]{enumitem}

\newcommand{\rar}{\rightarrow}

\newcommand{\bs}[1]{\boldsymbol{#1}}

\newcommand{\sign}{\operatorname{sign}}

\definecolor{deepgreen}{cmyk}{1,0,1,0.5}

\newcommand{\la}{\lambda}

\newcommand{\p}{\partial}

\makeatletter

\newcommand{\Rmnum}[1]{\expandafter\@slowromancap\romannumeral #1@}
\makeatother

\newcommand{\Del}[1]{}

\numberwithin{equation}{section}

\newtheorem{thm}{Theorem}[section]

\newtheorem{lem}[thm]{Lemma}
\newtheorem{prop}[thm]{Proposition}

\theoremstyle{remark}

\usepackage{tensor}

\usepackage{graphicx}
\usepackage{xcolor} 
\usepackage{tensor}
\usepackage{slashed}
\usepackage{cite}
\definecolor{green}{rgb}{0,0.8,0} 

\newcommand{\eps}{\epsilon}

\newcommand{\bbE}{\mathbb E}

\newcommand{\bbN}{\mathbb N}

\newcommand{\bbR}{\mathbb R}

\begin{document}

\title{On Stretch-Limited Elastic Strings}

\author{Casey Rodriguez}

\begin{abstract}
 Motivated by the increased interest in modeling nondissipative materials by constitutive relations more general than those from Cauchy elasticity, we initiate the study of a class of \emph{stretch-limited elastic strings}: the string cannot be compressed smaller than a certain length less than its natural length nor elongated larger than a certain length greater than its natural length. In particular, we consider equilibrium states for a string suspended between two points under the force of gravity (catenaries).  We study the locations of the supports resulting in tensile states containing both extensible and inextensible segments in two situations: the degenerate case when the string is vertical and the nondegenerate case when the supports are at the same height. We then study the existence and multiplicity of equilibrium states in general with multiplicity differing markedly from strings satisfying classical constitutive relations.    
\end{abstract}

\maketitle

\section{Introduction}
In a series of intriguing papers \cite{Raj_Implicit03, RajConspectus, RajElastElast, RajSmallStrain}, Rajagopal argued that elastic bodies should be defined as bodies incapable of dissipation and gave a wide class of implicit constitutive relations consistent with this definition. This class is larger than the constitutive relations found in standard Cauchy elasticity where stress is a function of strain.  In particular, it includes the simple case of the strain expressed as a function of the stress, which appears more consistent with Newtonian causality than standard constitutive relations since force causes displacement. Implicit constitutive relations have found a wide range of applications in the modeling of electro and magneto-elastic bodies \cite{BustRaj_Elasto13, BustRaj_Magneto15}, fracture in brittle materials \cite{GouMallRajWalton_PlaneCrack15, BulMalekRajWal_Existence15}, gum metal \cite{Raj_GumMetal14} and many other materials (see \cite{RajBust20} for further references).  

However, we are unaware of implicit constitutive relations being used in the theory of perfectly flexible one-dimensional elastic bodies moving in three-dimensional space: strings. The configuration at time $t$ of a finite-length elastic string is given by a nondegenerate curve in three-dimensional space, $[0,1] \ni s \mapsto \bs r(s,t)$, satisfying Newton's law:
\begin{align}
(\rho A)(s) \bs r_{tt}(s,t) = \bs n_s(s,t) + \bs f(s,t),
\end{align}
where $(\rho A)$ is the mass density, $\bs f$ is the body force density, $\bs n$ is the contact force given by 
\begin{align}
\bs n(s,t) = N(s,t) \frac{\bs r_s(s,t)}{|\bs r_s(s,t)|},
\end{align}
and $N$ is the tension. 
Similar to standard treatments of three-dimensional elastic bodies, in classical treatments of elastic strings, the system of equations are closed by expressing the tension $N$ as a function of the the stretch $\nu := |\bs r_s|$ via  
\begin{align}
N(s,t) = \hat N(\nu(s,t),s).
\end{align} 
Classically, it is assumed that the tension increases if the stretch increases (and thus, the relation between tension and stretch is invertible) with infinite compressive force associated to zero stretch and infinite tensile force associated to infinite stretch (see for example \cite{Antman79,AntmanBook} for more details). In particular, these relations allow the possibility of compressing a string to an arbitrarily small length or elongating a string to an arbitrarily large length.  In this work, we initiate the study of a new class of stretch-limited elastic strings. For these elastic strings, the stretch is expressed as a nondecreasing function of the tension,
\begin{align}
\nu(s,t) = \hat \nu(N(s,t),s), \quad \hat \nu(0,s) = 1, 
\end{align}
 but not vice versa, and the string cannot be compressed smaller than a certain length nor elongated larger than a certain length regardless of the magnitude of the tension: there exist $0 < \nu_0 < 1 < \nu_1$ such that for all $N \in \bbR, s \in [0,1]$,
 \begin{align}
 \hat \nu(N,s) \in [\nu_0,\nu_1].
 \end{align}
See Section 2 for the precise formulation and assumptions.

We consider a simple setting for stationary stretch-limited strings: a string suspended between two supports under the force of gravity (catenaries). The study of inextensible and extensible catenaries has a rich history going back to Galileo with contributions by the Bernoulli's, Huygens, Leibniz and many others (see \cite{Antman79} for a brief history). In this work, 
we first explicitly classify the support positions for which tensile catenaries ($N > 0$) contain inextensible segments (where the stretch is maximized) in two situations: the degenerate case when the string is vertical and the nondegenerate case when the supports are at the same height (see Proposition \ref{p:elong_threshold} for the degenerate case, and Proposition \ref{p:inext_cond} and Proposition \ref{p:ext_inext_condition} for the nondegenerate case). We then turn to the study of existence and multiplicity of stationary states for given points of support. It is simple to show that tensile states exist and are unique as in the case of standard constitutive relations considered in \cite{Antman79}. In the same work, for standard constitutive relations, Antman also proved that given two points of support less than one unit of reference length apart, there exist multiple compressed states ($N < 0$) as long as the mass density is small. In contrast, for stretch-limited strings we prove that if the distance between supports is less than the minimal length of the string then there exists a unique compressed state, as long as the mass density is small (see Proposition \ref{p:uniqueness}).  We comment that although compressed states are unstable within the theory of elastic strings, they may play a more prominent role in a theory of rods with implicit constitutive relations (a topic we intend to pursue) where their concave graphs model moment-free arches. 

\section{Stretch-Limited Strings}

\subsection{Kinematics, equations of motion, and constitutive assumptions}

Our general formulation of elastic strings follows the standard treatment due to Antman \cite{Antman79,AntmanBook}. Let $\{ \bs i, \bs j, \bs k \}$ be a fixed right-handed orthonormal basis for Euclidean space $\bbE^3$. Let $s\in  I$, where $I$ is the interval $[0,1], [0,\infty),$ or $(-\infty,\infty)$, which parameterizes the \emph{material points} of the string. The \emph{configuration} of the string at time $t$ is the map $s \mapsto \bs r(s,t)$, and the tangent vector to the curve $\bs r(\cdot, t)$ at $s$ is $\bs r_s(s,t) = \p_s \bs r(s,t)$. We recall that the \emph{stretch} $\nu(s,t)$ of the string at $(s,t)$ is 
\begin{align*}
\nu(s,t) := |\bs r_s(s,t)|.
\end{align*}
We require that the stretch is always positive throughout the motion of the string and say the string is \emph{elongated} where $\nu(s,t) > 1$ and \emph{compressed} where $\nu(s,t) < 1$. 

As a result of balance of linear momentum, we have the \emph{classical equations of motion} for a string: 
\begin{align}
(\rho A)(s) \bs r_{tt}(s,t) = \bs n_s(s,t) + \bs f(s,t), \quad (s,t) \in I \times (0,\infty). \label{eq:motion}
\end{align}
Here $(\rho A)(s)$ is the mass per unit reference length at $s$, $\bs f(s,t)$ is the body force per unit reference length at $(s,t)$, and $\bs n(s,t)$ is the contact force at $(s,t)$. 

A defining characteristic of a string is that the contact force is assumed to be tangent to the configuration of the string. Thus, there exists a scalar-valued function $N(s,t)$, the \emph{tension}, such that 
\begin{align*}
\bs n(s,t) = N(s,t) \frac{\bs r_s(s,t)}{|\bs r_s(s,t)|}
\end{align*}  
The mechanical properties of a string are modeled by specifying a relation between the stretch $\nu$ and tension $N$. In standard treatments of mechanical strings,  a string is said to be \emph{elastic} if there exists a function $\hat N(\nu,s)$ such that 
\begin{align}
N(s,t) = \hat N(\nu(s,t), s) \label{eq:constN}
\end{align}
(see \cite{Antman79, AntmanBook}). Inserting this relation into the equations of motion \eqref{eq:motion} results in a closed system of partial differential equations for the variable $\bs r(s,t)$ (which are hyperbolic if $\hat N_\nu \geq c > 0$ and the tension is positive, $N > 0$). A string is \emph{inextensible} if $\nu = 1$ no matter the force applied. In this case, the contact force $\bs n(s,t)$ is determined by \eqref{eq:motion} and the condition $\nu = 1$ rather than a constitutive relation. 

 Physically reasonable assumptions imposed on $\hat N$ are that an unstretched configuration is not in a state of tension, an increase in tension leads to an increase in stretch, a state of zero stretch requires infinite compression force and a state of infinite stretch requires infinite tensile force.  Mathematically, these assumptions are expressed by assuming $\hat N(1,s) = 0$, $\nu \mapsto \hat N(\nu,s)$ is increasing, $\lim_{\nu \rar 0} \hat N(\nu,s) = -\infty$ and $\lim_{\nu \rightarrow \infty} \hat N(\nu,s) = \infty$ for all $s$. Thus, $\hat N(\cdot,s)$ has an inverse function $\hat \nu(\cdot, s)$, and the constitutive assumption takes the form 
\begin{align}
\nu(s,t) = \hat \nu(N(s,t),s), \label{eq:const}
\end{align}
where 
$N \mapsto \hat \nu(N,s) \mbox{ is increasing,}$ \\
\begin{align} 
\begin{split}
\lim_{N \rar -\infty} \hat \nu(N,s) &= 0, \\
\lim_{N \rar \infty} \hat \nu(N,s) &= \infty.  
\end{split}\label{eq:ext_const}
\end{align}

\subsection{A class of stretch-limiting constitutive relations}

Motivated by the intriguing papers by Rajagopal \cite{Raj_Implicit03, RajConspectus, RajElastElast, RajSmallStrain}, we consider a class of strings with constitutive relation expressed via \eqref{eq:const} which cannot be expressed via \eqref{eq:constN}. In particular, we assume that the string cannot be further elongated nor compressed once a threshold value of tensile force or compression force is reached. We refer to the modeled objects as \emph{stretch-limited strings}. Mathematically, we assume that 
\begin{itemize}
\item there exist $N_0 < 0 < N_1$ and $\nu_0 < 1 < \nu_1$ such that 
\begin{align*}
\hat \nu(N,s) = 
\begin{cases}
\nu_0 &\mbox{ if } N \leq N_0, \\
\nu_1 &\mbox{ if } N \geq N_1, 
\end{cases}
\end{align*}
\item the function $\hat \nu(\cdot, s) \in C^\infty([N_0,N_1]; [\nu_0,\nu_1])$, $\hat \nu(0,s) = 1$ and there exists $c > 0$ such that $\hat \nu_N(N,s) \geq c$ for all $(N,s) \in [N_0,N_1] \times I$. 
\end{itemize}
We note that the previous two assumptions imply that $\hat \nu(\cdot,s)$ is continuous and piece-wise smooth but not globally smooth. The second assumption is necessary for the equations of motion \eqref{eq:motion} to be hyperbolic in segments of the string not fully stretched. 

Due to our assumptions on the constitutive relation between stretch and tension, one cannot expect motions in which the string has both extensible and inextensible segments to be classical solutions to \eqref{eq:motion} across the interface. If the interface between the extensible segment of the string $(N \in [N_0, N_1])$ and the inextensible segment of the string $(N \in (-\infty, N_0] \cup [N_1,\infty))$ is given by a curve $(\sigma(t),t)$ in the $(s,t)$ plane, then along the interface, $N \in \{N_0,N_1\}$, a weak solution to \eqref{eq:motion} satisfies the well-known \emph{Rankine-Hugoniot jump conditions}: 
\begin{align}
\llbracket \bs n \rrbracket + (\rho A) \sigma' \llbracket \bs r_t \rrbracket = \bs 0. 
\end{align} 
where $\llbracket \bs y \rrbracket(\sigma(t),t) = \bs y(\sigma(t)^+,t) - \bs y(\sigma(t)^-,t)$ is the jump across the point $s = \sigma(t)$ at time $t$. In particular, for equilibrium states we must have $\bs n$ is continuous across the interface. 

\section{Stretch-Limited Vertical States} 

\subsection{Formulation and initial result}
We first consider the degenerate catenary problem for a straight vertical state: 
\begin{align}
&\bs n(s) =\bs n(0) + F(s) \bs k, \quad s \in [0,1], \\
&\bs r(s) = z(s) \bs k, \quad \bs r(0) = \bs 0, \quad \bs r(1) = b \bs k, \label{eq:stat_straight}
\end{align} 
where $b > 0$, $\bs n(s) = N(s)\frac{\bs r_s(s)}{|\bs r_s(s)|}$ is the contact force, the stretch $\nu(s) = |\bs r_s(s)|$ and tension $N(s)$ satisfy the constitutive relation discussed in Section 2.2,
and the magnitude of the total gravitational force on the material segment $[0,s]$ is 
\begin{align}
F(s) := \int_0^s g (\rho A)(\xi) d\xi. 
\end{align} 
We denote the total mass of the string by $m := \int_0^1 (\rho A)(s) ds$.

\begin{prop}
Let $b \in [\nu_0,\nu_1]$. Then there exists $N(0) \in \bbR$ such that 
\begin{align*}
N(s) &:= N(0) + F(s), \\ 
\bs n(s) &:= N(s) \bs k, \\
\bs r(s) &:= \int_0^s \hat \nu(N(\xi),\xi) d\xi \bs k,
\end{align*} for $s \in [0,1]$ solve \eqref{eq:stat_straight}. If $b \in (\nu_0, \nu_1)$ then $N(0)$ is unique (and thus, $\bs r$ is unique). If $b = \nu_0$ or $b = \nu_1$, then $\bs r$ is unique. 
\end{prop}

\begin{proof}
We write $\bs r(s) = z(s) \bs k$ so that $\nu(s) = z'(s) > 0$ on $[0,1]$. 
	If $z(0) = 0$ and $z(1) = b$ then 
	\begin{align*}
	b = \int_0^1 z'(s) ds = \int_0^1 \hat \nu(N(s)) ds 
	\end{align*}
	so \eqref{eq:stat_straight} is equivalent to
	\begin{align}
	N(s) &= N(0) + F(s), \\
	b &= \int_0^1 \hat \nu(N(0) + F(s),s) ds \label{eq:L_equation}
	\end{align}
	Since $F(s)$ is increasing on $[0,1]$, $F(1) = gm$ and $\bs r$ is stretch-limited, we have 
	\begin{align*}
N(0) \leq N_0 - g m \iff \int_0^1 \hat \nu(N(0) + F(s),s) ds = \nu_0, \\
N(0) \geq N_1 \iff \int_0^1 \hat \nu(N(0) + F(s),s) ds = \nu_1. 
	\end{align*} 
	Since the function $N(0) \mapsto \int_0^1 \hat \nu(N(0) + F(s),s) ds$ is continuous and increasing on $[N_0-gm, N_1]$, the proposition follows from the intermediate value theorem. 
\end{proof}

\subsection{Threshold for elongated mixed extensible-inextensible state} 

We now assume that the string is uniform so that 
\begin{align}
(\rho A)(s) = \gamma >  0, \quad \nu(s) = \hat \nu (N(s)).
\end{align} Since $F(s) = g \gamma s$ is increasing, there exist equilibrium states  given by a union of an elongated extensible segment where $\nu \in (1,\nu_1)$ and an inextensible segment where $\nu = \nu_1$. Indeed, this occurs if and only if 
\begin{align}
N(0) >  0 \quad \mbox{and} \quad \frac{N_1 - N(0)}{g\gamma} \in (0,1). \label{eq:exin_condition}
\end{align} 
Indeed, \eqref{eq:exin_condition} is equivalent to
$0 < N(s) < N_1$ for all $s \in [0,s_1)$ and $N(s) \geq N_1$ for all $s \in [s_1,1]$ where 
\begin{align*}
s_1 := \frac{N_1 - N(0)}{g \gamma} \in (0,1).
\end{align*} 

In terms of the position of the support $b \bs k$, we pose the following question: 
\begin{itemize}
	\item  What values of $b$ result in an elongated state composed of an extensible segment and inextensible segment? 
\end{itemize}
We prove that the essentially sharp threshold for $b$ is 
$
b(\gamma) = \nu_1 - \frac{1}{2} g \gamma \hat \nu_{N^-}(N_1),
$
for all $\gamma$ sufficiently small. 

\begin{prop}\label{p:elong_threshold}
Let $\epsilon \in (0,1)$. Then for all $\gamma$ sufficiently small (depending on $\eps$ and $\hat \nu$), the following is true. If 
\begin{align}
\nu_1 - \Bigl (1-\eps \Bigr ) \frac{g\gamma}{2} \hat \nu_{N^-}(N_1) \leq b < \nu_1, \label{eq:L_ass}
\end{align}
then \eqref{eq:exin_condition} holds, i.e.
the string is a union of an elongated extensible segment and inextensible segment.

Conversely, if the string is a union of an elongated extensible segment and inextensible segment, then
\begin{align}
\nu_1 - (1 + \eps) \frac{g\gamma}{2} \hat \nu_{N^-}(N_1) \leq b < \nu_1, \label{eq:L_assc}
\end{align} 
\end{prop} 

\begin{proof}
Suppose that \eqref{eq:L_ass} holds. Since $b < \nu_1$ and $N(s) = N(0) + g\gamma s$, we must have $N(0) < N_1$. Moreover, since $b > \frac{\nu_1}{2}$ for all $\gamma$ sufficiently small, it follows that $N(0) > 0$. Otherwise, for all $s$, $N(s) \leq \gamma g$ which implies $b \leq \hat \nu(\gamma g) < \frac{\nu_1}{2}$ for all $\gamma$ sufficiently small. Thus, $N(0) \in (0,N_1)$

We now prove the conclusion of the first part of Proposition \ref{p:elong_threshold} by contradiction. Suppose that there exist $b_n > 0$, $\gamma_n >  0$ and $0 < N_n(0) < N_1$ satisfying \eqref{eq:L_ass}, $\gamma_n \rar 0$ and for all $n$, 
\begin{align}
\frac{N_1 - N_n(0)}{g\gamma_n} \geq 1. \label{eq:contr}
\end{align}
Let $N_n(s) := N_n(0) + g \gamma_n s$. 
By Taylor's theorem we have
\begin{align}
b &= \int_0^1 \hat \nu(N_n(0) + g \gamma_n s) ds \\
&= \hat \nu(N_n(0)) + O(\gamma_n),\label{eq:LN0}
\end{align}
so by \eqref{eq:L_ass}, $\nu_1 - \hat \nu(N_n(0)) = O(\gamma_n)$. 
Since $\nu_1 = \hat \nu(N_1)$ and the function $\hat \nu$ is smoothly invertible on 
$[N_0,N_1]$, we conclude 
$
|N_1 - N_n(0)| = O(\gamma_n), 
$
and thus $
|N_1 - N_n(s)| = O(\gamma_n)
$
uniformly in $s$. 

Again by Taylor's theorem 
\begin{align*}
b &= \int_0^1 \hat \nu(N_n(s)) ds \\
&= \int_0^1 \left [
\hat \nu(N_1) + \hat \nu_{N^-}(N_1)(N_n(s) - N_1) \right ] ds + O(\gamma_n^2) \\
&= \nu_1 + \hat \nu_{N^-}(N_1)\left ( N_n(0) - N_1 + \frac{1}{2}g\gamma_n \right ) + O(\gamma_n^2). 
\end{align*}
Thus, 
\begin{align*}
\frac{N_1 - N_n(0)}{g\gamma_n} = \frac{\nu_1 - b}{g \gamma \hat \nu_{N^-}(N_1)} + \frac{1}{2} + O(\gamma_n) \leq (1-\eps) \frac{1}{2} + \frac{1}{2} + O(\gamma_n) < 
1
\end{align*}
for all $n$ sufficiently large. This contradicts \eqref{eq:contr}, and thus, $[N_1 - N(0)]/g\gamma \in (0,1)$. 

Now suppose that $N(0) > 0$ and $s_0 := [N_1 - N(0)]/g\gamma \in (0,1)$. Then as above, we use Taylor's theorem and the fact $\hat \nu_{N^-}(N_1) \geq c > 0$ to expand 
\begin{align*}
b &= \int_0^1 \hat \nu (N(s)) ds \\
&= \int_0^{s_0}\hat \nu(N(s)) ds + \int_{s_0}^1 \nu_1 ds \\
&= \int_0^{s_0} \left [
\hat \nu(N_1) + \hat \nu_{N^-}(N_1)(1 + O(N(s)-N_1))(N(s) - N_1) \right ] ds \\
&\quad+ (1 - s_0)\nu_1 \\
&= s_0 \nu_1 - \frac{1}{2} g\gamma \hat \nu_{N^-}(N_1)(1 + O(\gamma)) s_0^2 + (1 - s_0) \nu_1 \\
&= \nu_1 - \frac{1}{2}g \gamma \hat \nu_{N^-}(N_1) (1 + O(\gamma))s_0^2. 
\end{align*}
Thus, 
\begin{align*}
\frac{\nu_1 - b}{\frac{1}{2}g\gamma \hat \nu_{N^-}(N_1) (1+O(\gamma))} = s_0^2 \in (0,1)
\end{align*}
proving \eqref{eq:L_assc}.
\end{proof}

\section{Stretch-Limited Catenaries}

\subsection{Formulation and preliminary results}
Suppose $a > 0, b \in \bbR$. We now consider the nondegenerate catenary problem 
\begin{align}
&\bs n(s) = \bs n(0) + F(s) \bs k, \\
&\bs r(0) = \bs 0, \quad \bs r(1) = a \bs i + b \bs k, \label{eq:catenary}
\end{align}
where, as before, $\bs n(s) = N(s)\frac{\bs r_s(s)}{|\bs r_s(s)|}$ is the contact force, the stretch $\nu(s) = |\bs r_s(s)|$ and tension $N(s)$ satisfy the constitutive relation discussed in Section 2.2,
and the magnitude of the total gravitational force on the material segment $[0,s]$ is 
\begin{align}
F(s) = \int_0^s g (\rho A)(\xi) d\xi. 
\end{align} 

It is well-known that the assumptions of the problem imply that the configuration is planar, $\bs r(s) \in \mbox{span} \{ \bs i, \bs k\}$, $\bs n(0) \cdot \bs i \neq 0$, and $\bs n$ is nowhere vanishing (see \cite{Antman79,AntmanBook}).

Writing 
\begin{align}
\frac{\bs r}{|\bs r|} &= \cos \theta \bs i + \sin \theta \bs k, \\
\bs n(0) &= \la \bs i + \mu \bs k,
\end{align} 
we obtain from \eqref{eq:catenary} the relations 
\begin{gather}
N \cos \theta = \la, \quad N \sin \theta = \mu + F, \\
N = \la \cos \theta + (\mu + F) \sin \theta, 
\end{gather}
and thus, if $\delta := \sqrt{\la^2 + (\mu + F)^2}$, 
\begin{gather}
\tan \theta = \frac{\mu + F}{\la},\\
N = \sign (\la) \delta, \quad \cos \theta = \frac{|\la|}{\delta}, \quad 
\sin \theta = \sign(\la) \frac{\mu + F}{\delta}. 
\end{gather}
With $a > 0$ and $b$ specified, the relation between $(a,b)$ and $(\la,\nu)$ is given by 
\begin{align}
a \bs i + b \bs k &= \int_0^1 \bs r_s(s) ds \\
&= \int_0^1 \nu(s) [ \cos \theta(s) \bs i + \sin \theta(s) \bs k] ds \\
&= \int_0^1 \frac{\hat \nu(\pm \delta(s),s)}{\pm \delta(s)}
[\la \bs i + (\mu + F(s)) \bs k ] ds \\
&= P^{\pm}(\la, \mu) \bs i + Q^{\pm}(\la, \mu) \bs k, \label{eq:PQeq}
\end{align}
where the $\pm$ corresponds to the sign of $\la$. 
This last relation can be written as 
\begin{align}
\nabla_{\la,\mu} \Phi^{\pm}(\la,\mu) = \bs 0, \label{eq:phiequ}
\end{align}
where 
\begin{align}
\Phi^{\pm}(\la, \mu) &:= \int_0^1 W^*(\pm \delta(s),s) ds - \la a - \mu b, \\
W^*(N,s) &:= \int_0^N \hat \nu(\bar N,s) d\bar N. 
\end{align}
We note that since the stretch $\nu(s) \leq \nu_1$ for all $s$, any solution to \eqref{eq:PQeq} (and thus \eqref{eq:phiequ}) must satisfy 
\begin{align*}
a^2 + b^2 < \nu_1^2.
\end{align*}

By a simple adaption of the proof from \cite{Antman79} using the variational form of the problem \eqref{eq:phiequ}, we have the following existence and uniqueness result for tensile states $(\la >  0)$ and existence result for compressive states $(\la < 0$). 

\begin{prop}
Suppose $a > 0$ and $a^2 + b^2 < \nu_1^2$. Then there exists unique $(\la^+,\mu^+)$ with $\la^+ > 0$ satisfying $\nabla \Phi^+(\la^+,\mu^+) = 0$. 

Suppose $a > 0$ and $a^2 + b^2 < \nu_0^2$. Then there exists $(\la^-,\mu^-)$ with $\la^- < 0$ satisfying $\nabla \Phi^-(\la^-,\mu^-) = 0$.

\end{prop}

\subsection{Threshold for tensile states containing an inextensible segment}
For the remainder of the paper we assume that the catenary is uniform so that 
\begin{align}
(\rho A)(s) = \gamma > 0, \quad \nu(s) = \hat \nu(N(s)). 
\end{align}
Suppose that $b = 0$ (the supports of the catenary are at the same height) so the tensile catenary satisfies 
\begin{align}
a &= \int_0^1 \hat \nu(\delta(s)) \frac{\la}{\delta(s)} ds, \label{eq:aeq}\\
0 &= \int_0^1 \hat \nu(\delta(s)) \frac{\mu + g \gamma s}{\delta(s)} ds. \label{eq:beq}
\end{align} 
Then \eqref{eq:beq} implies that $\mu = -\frac{1}{2} g \gamma$ and \eqref{eq:aeq} becomes 
\begin{align}
a &= \int_0^1 \hat \nu(\delta(s)) \frac{\la}{\delta(s)} ds,
\quad \delta(s) = \sqrt{\la^2 + g^2 \gamma^2\Bigl (s - \frac{1}{2} 
	\Bigr )^2} \label{eq:a_equation}.
\end{align}

A tensile catenary is inextensible ($\nu = \nu_1$) if and only if for all $s \in [0,1]$, $\delta(s) \geq N_1$ which is equivalent to $\la \geq N_1$. 
A tensile catenary is a union of extensible $(\nu \in (1,\nu_1))$ and inextensible $(\nu = \nu_1)$ segments if and only if $\la < N_1$ and there exists $s \in (0,1)$ such that $\delta(s) = N_1$, which is equivalent to
\begin{align}
\frac{1}{(g \gamma)^2}(N_1^2 - \la^2) \in (0,1/4). \label{eq:inext_ext_condition}
\end{align}
The condition \eqref{eq:inext_ext_condition} is equivalent to $\delta(s) \geq N_1$ for all $s \in [0,s_-] \cup [s_+,1]$ and $0 < \delta(s) < N_1$ for all $s \in (s_-,s_+)$ where 
\begin{align}
s_{\pm} = \frac{1}{2} \pm \frac{1}{g \gamma}\sqrt{N_1^2 - \la^2} \in (0,1). 
\end{align}
Thus, a tensile catenary containing an inextensible segment is either completely inextensible or is a union of an extensible segment and two inextensible segments. 

In terms of the support at $a \bs i$, we now pose a similar question as in Section 3: 
\begin{itemize}
	\item  What values of $a$ result in an tensile state containing inextensible segments? 
\end{itemize}

We answer this question explicitly in the following two propositions. 

\begin{prop}\label{p:inext_cond}
For all $\gamma$, the tensile catenary is inextensible if and only if 
\begin{align}
\nu_1 \frac{2N_1}{g \gamma} \sinh^{-1} \frac{g \gamma}{2 N_1} \leq a < \nu_1. \label{eq:inext_condition}
\end{align}
\end{prop}

\begin{proof}
If the catenary is inextensible, then $\la \geq N_1$ and 
\begin{align}
a = \int_0^1 \nu_1 \frac{\la}{\delta(s)} ds =
\nu_1 \frac{2\la}{g \gamma} \sinh^{-1} \frac{g \gamma}{2 \la} \geq 
\nu_1 \frac{2N_1}{g \gamma} \sinh^{-1} \frac{g \gamma}{2 N_1} 
\end{align}	
since the function $z \mapsto \frac{\sinh^{-1} z}{z}$ is decreasing. 

Conversely, if \eqref{eq:inext_condition} holds, then the unique $\la \geq N_1$ solving 
$\frac{2\la}{g \gamma} \sinh^{-1} \frac{g \gamma}{2 \la} = \frac{a}{\nu_1} \in (0,1)$ satisfies the equation $a = \int_0^1 \nu_1 \frac{\la}{\delta(s)} ds$ by the previous computation. By the uniqueness result for tensile states, it follows that the catenary is inextensible. 
\end{proof}

\begin{prop}\label{p:ext_inext_condition}
Let $\eps \in (0,1)$. For all $\gamma$ sufficiently small (depending on $\eps$ and $\hat \nu$), the following is true. If 
\begin{align}
\nu_1 - \frac{(g \gamma)^2}{24 N_1^2}
(\nu_1 + (2+3\eps) \hat \nu_{N^-}(N_1) N_1)
 \leq a
< \nu_1 \frac{2N_1}{g \gamma} \sinh^{-1} \frac{g \gamma}{2 N_1}, \label{eq:position_ass}
\end{align}
then \eqref{eq:inext_ext_condition} holds i.e. the tensile catenary is a union of an extensible segment and two inextensible segments.  

Conversely, if the tensile catenary is composed of an extensible segment and two inextensible segments, then 
\begin{align}
\nu_1 - \frac{(g \gamma)^2}{24 N_1^2}
(\nu_1 + (2-3\eps) \hat \nu_{N^-}(N_1) N_1)
\leq a
< \nu_1 \frac{2N_1}{g \gamma} \sinh^{-1} \frac{g \gamma}{2 N_1}, \label{eq:position_ass_conv}.
\end{align}
\end{prop}

\begin{proof}
We prove the first part of the proposition by contradiction. Suppose that there exist $a_n > 0$, $\gamma_n > 0$ and $\la_n > 0$ such that $\gamma_n \rar 0$ and either 
\begin{align}
\forall n, \quad \la_n \geq N_1, \label{eq:in_ext_non}
\end{align}
or 
\begin{align}
\forall n, \quad \frac{1}{(g \gamma)^2}(N_1^2 - \la_n^2) \geq \frac{1}{4}
. \label{eq:inext_ext_condition_non}
\end{align}
By \eqref{eq:position_ass} and Proposition \ref{p:ext_inext_condition}, the possibility \eqref{eq:in_ext_non} is immediately ruled out. Assume that $\forall n$, $\la_n < N_1$ and \eqref{eq:inext_ext_condition_non} holds. We claim that $\la_n \rar N_1$. If not then there exists a subsequence $\{\la_{n_k} \}_k$ and $\delta_1 \in [0,N_1)$ such that $\la_{n_k} \rar \delta_1$. Then by \eqref{eq:position_ass}
\begin{align*}
\nu_1 = \lim_{k \rar \infty} a_{n_k} \leq \lim_{k \rar \infty} \int_0^1 \hat \nu(\delta_{n_k}(s)) ds = \hat \nu(\delta_1) < \nu_1,
\end{align*}
a contradiction. Thus, $\la_n \rar N_1$. In particular, for all $n$ sufficiently large, $\la_n \geq N_1/2 > 0$. Now 
\begin{align}
\delta_n(s) - \la_n = O(\gamma_n^2),
\end{align}
uniformly in $s$, 
so by Taylor's theorem 
\begin{align}
a &= \int_0^1 \hat \nu(\delta_n(s)) \frac{\la_n}{\delta_n(s)} ds \\
&= \hat \nu(\la_n) \frac{2 \la_n}{g \gamma_n} \sinh^{-1} \frac{g \gamma_n}{2 \la_n}
+ O(\gamma_n^2) \\
&= \hat \nu(\la_n) + O(\gamma_n^2). 
\end{align}
By \eqref{eq:position_ass} we conclude 
\begin{align}
|\nu_1 - \nu(\la_n)| = O(\gamma_n^2). 
\end{align}
Since $\hat \nu$ is smoothly invertible on $[N_0,N_1]$, we conclude $N_1 - \la_n = O(\gamma_n^2)$ and thus, 
\begin{align}
N_1 - \delta_n(s) = O(\gamma_n^2).  
\end{align}
We use Taylor's theorem again 
\begin{align}
a &= \int_0^1 \hat \nu(\delta_n(s)) \frac{\la_n}{\delta_n(s)} ds \\
&= \int_0^1 \left (
\nu_1 + \hat \nu_{N^-}(N_1)(\delta_n(s)-N_1) + O(\gamma_n^4)
\right ) \frac{\la_n}{\delta_n(s)} ds \\
&= \left (
\nu_1 - \hat \nu_{N^-}(N_1) N_1
\right ) \frac{2\la_n}{g \gamma_n} \sinh^{-1} \frac{g \gamma_n}{2 \la_n} + 
\hat \nu_{N^-}(N_1) \la_n + O(\gamma_n^4) \\
&= \nu_1 - \hat \nu_{N^-}(N_1) N_1 + \hat \nu_{N^-}(N_1) \la_n 
- \frac{1}{24 N_1^2} (\nu_1 - \hat \nu_{N^-}(N_1) N_1) + O(\gamma_n^4)
\end{align}
and thus
\begin{align}
\frac{N_1}{g \gamma_n} = \frac{\la_n}{g \gamma_n}
+ \frac{g \gamma_n}{\hat \nu_{N^-}(N_1)} \left (
\frac{\nu_1 - a}{(g \gamma_n)^2} - \frac{1}{24 N_1^2}(\nu_1 - \hat \nu_{N^-}(N_1) N_1)
\right ) + O(\gamma_n^3).
\end{align}
We conclude 
\begin{align}
\frac{1}{4} &\leq \frac{N_1^2 - \la_n^2}{(g \gamma_n)^2} \\
&= \frac{2 \la_n}{\hat \nu_{N^-}(N_1)} \left (
\frac{\nu_1 - a}{(g \gamma_n)^2} - \frac{1}{24 N_1^2}(\nu_1 - \hat \nu_{N^-}(N_1) N_1)
\right ) + O(\gamma_n^2) \\
&= \frac{2 N_1}{\hat \nu_{N^-}(N_1)} \left (
\frac{\nu_1 - a}{(g \gamma_n)^2} - \frac{1}{24 N_1^2}(\nu_1 - \hat \nu_{N^-}(N_1) N_1)
\right ) + O(\gamma_n^2) \\
&\leq (1 - \eps)\frac{1}{4} + O(\gamma_n^2) < \frac{1}{4},
\end{align}
for all $n$ sufficiently large. This contradiction shows \eqref{eq:inext_ext_condition_non} cannot hold and proves the first part of the proposition. 

Now assume that the tensile catenary is the union of an extensible segment and two inextensible segments: we have 
\begin{align}
\tau := \frac{1}{g\gamma} \sqrt{N_1^2 - \la^2} \in (0,1/2)
\end{align}
so the segments $[0,1/2 - \tau]$ and $[1/2+\tau,1]$ are inextensible, $\delta \geq N_1$, and the segment $(1/2-\tau,1/2+\tau)$ is extensible, $\delta \in (N_0,N_1)$.
Then $|N_1 - \la| \leq (g \gamma)^2/4N_1$ so $|\delta(s) - N_1| = O(\gamma^2)$ uniformly in $s \in [0,1]$. We then deduce the relation 
\begin{align}
\la = N_1 - \frac{(g \gamma)^2 \tau^2}{2 N_1} + O(\gamma^4). 
\end{align}

Applying Taylor's theorem we have 
\begin{align}
a &= \int_0^1 \hat \nu(\delta(s)) \frac{\la}{\delta(s)} ds \\
&= \int_0^{1/2-\tau} \nu_1 \frac{\la}{\delta(s)} ds + 
\int_{1/2-\tau}^{1/2+\tau} \hat \nu(\delta(s)) \frac{\la}{\delta(s)} ds + 
\int_{1/2+\tau}^1 \nu_1 \frac{\la}{\delta(s)} ds \\
&= \int_{1/2-\tau}^{1/2 + \tau} (\nu_1 + \hat \nu_{N^-}(N_1)(\delta(s) - N_1) + O(\gamma^4)) \frac{\la}{\delta(s)} ds \\
&\quad+ \frac{2\la \nu_1}{g \gamma} \sinh^{-1} \frac{g \gamma}{2\la} 
- \frac{2 \la \nu_1}{g \gamma} \sinh^{-1} \frac{g \gamma \tau}{\la} \\
&= (\nu_1 - \hat \nu_{N^-}(N_1)N_1) \frac{2 \la}{g \gamma} \sinh^{-1} \frac{g \gamma \tau}{\la} + 2 \tau \hat \nu_{N^-}(N_1) \la + O(\gamma^4) \\
&\quad+ \frac{2\la \nu_1}{g \gamma} \sinh^{-1} \frac{g \gamma}{2\la} 
- \frac{2 \la \nu_1}{g \gamma} \sinh^{-1} \frac{g \gamma \tau}{\la} \\
&= -\hat \nu_{N^-}(N_1)N_1\Bigl (2 \tau - \frac{(g \gamma)^2 \tau^3}{3\la^2} \Bigr ) + 2 \tau \hat \nu_{N^-}(N_1) \la + \nu_1 - \frac{\nu_1(g \gamma)^2}{24 \la^2} + O(\gamma^4) \\
&= -\hat \nu_{N^-}(N_1)N_1\Bigl (2 \tau - \frac{(g \gamma)^2 \tau^3}{3N_1^2} \Bigr ) + 2 \tau \hat \nu_{N^-}(N_1) \la + \nu_1 - \frac{\nu_1(g \gamma)^2}{24 N_1^2} + O(\gamma^4).
\end{align}
Thus, 
\begin{align}
\nu_1 - &\frac{(g \gamma)^2}{24 N_1^2} (\nu_1 -\hat \nu_{N^-}(N_1)N_1 ) \\
&= a + \tau \hat \nu_{N^-}(N_1) N_1 \Bigl (
2 - \frac{(g \gamma)^2 \tau^2}{3 N_1^2} - \frac{2 \la}{N_1}
\Bigr ) + \frac{(g\gamma)^2}{24 N_1} \hat \nu_{N^-}(N_1) + O(\gamma^4) \\
&= a + \tau \hat \nu_{N^-}(N_1) \frac{1}{3N_1} \Bigl (
6 N_1^2 - (g \gamma)^2 \tau^2- 6 \la N_1
\Bigr ) \\&\quad+ \frac{(g\gamma)^2}{24 N_1} \hat \nu_{N^-}(N_1)  + O(\gamma^4) \\
&= a + (g \gamma)^2 \frac{\hat \nu_{N^-}(N_1)}{N_1} \Bigl ( \frac{2 \tau^3}{3} + \frac{1}{24} \Bigr ) + O(\gamma^4) \\
&< a + (g \gamma)^2 \frac{\hat \nu_{N^-}(N_1)}{8 N_1} + O(\gamma^4) \\
&<  a + (g \gamma)^2 \frac{\hat \nu_{N^-}(N_1)}{8 N_1}(1 + \eps) 
\end{align}
for all $\gamma$ sufficiently small, since $\hat \nu_{N^-}(N_1) \geq c > 0$. This proves \eqref{eq:position_ass_conv} and concludes the proof of the proposition. 
\end{proof}

\subsection{Existence and multiplicity of compressive states}

We now consider the existence and multiplicity of compressive states. Since $\delta(s) \neq 0$ for all $s \in [0,1]$, compressive states satisfy 
$a^2 + b^2 < 1$. 
In \cite{Antman79}, Antman proved that if $\nu = \hat \nu(N)$ satisfies \eqref{eq:ext_const}, $a > 0$  and $a^2 + b^2 < 1$, then for all $\gamma$ sufficiently small, there are \emph{at least two} solution pairs $(\la,\mu)$ to
\begin{align}
a \bs i + b \bs k
= \int_0^1 \frac{\hat \nu(-\delta(s))}{-\delta(s)}
\bigl ( \la \bs i + (\mu + g\gamma s) \bs k \bigr ) ds. \label{eq:ext_problem}
\end{align}
In \cite{Wolfe97}, Wolfe proved that of these multiple states, one state is a perturbation of a unit massive inextensible catenary, and another state is a perturbation of a straight, mass-less catenary. A simple adaptation of \cite{Antman79} yields the following result in our stretch-limited setting.
 
\begin{prop}\label{p:slightly_ext}\mbox{}
Suppose that $a > 0$ and $\nu_0^2 < a^2 + b^2 < 1$. Then for all $\gamma$ sufficiently small, there exist at least two solution pairs to \eqref{eq:ext_problem}. 
\end{prop}
	
The question we turn to for the remainder of the section is: 
\begin{itemize}
 \item What is the multiplicity of solutions to \eqref{eq:ext_problem} when the distance between supports is less than the minimal length of the string, i.e. $a^2 + b^2 < \nu_0^2$? Moreover, do these solutions contain inextensible segments where $\nu = \nu_0$?  
\end{itemize}
We prove that if $a > 0$ and $a^2 + b^2 < \nu_0^2$, then a solution pair $(\la,\mu)$ to \eqref{eq:ext_problem} is unique, and the catenary is completely extensible, as long as the mass density is small. 

We recall that if $a > 0$ and $a^2 + b^2 < 1$, then the problem for a uniform compressive inextensible catenary with $(\rho A)(s) = \gamma_0 > 0$, $\nu = 1$,
\begin{align}
a \bs i + b \bs k
= \int_0^1 \frac{1}{-\delta(s)}
\bigl ( \la \bs i + (\mu + g\gamma_0 s) \bs k \bigr ) ds, \label{eq:inext}
\end{align}
can be solved explicitly using hyperbolic functions (see \cite{Antman79}). In particular, $(\la, \mu)$ solving \eqref{eq:inext} are uniquely determined by the relations: 
\begin{align}
\frac{\sqrt{1 - b^2}}{a} &= \frac{2 |\lambda|}{a g \gamma_0} \sinh \frac{a g \gamma_0}{2 |\la|}, 
\label{eq:lambda_equation} \\
\mu &= \lambda \sinh \left (
\frac{a g \gamma_0}{2 |\lambda|} + \tanh^{-1} b
\right ) \label{eq:mu_equation}. 
\end{align}
We note that \eqref{eq:lambda_equation} uniquely determines $\lambda$ since the function $z \mapsto \frac{\sinh z}{z}$ is invertible on $(0,\infty)$ with range $(1,\infty)$ and $\frac{\sqrt{1 - b^2}}{a} > 1$.

\begin{prop}\label{p:uniqueness}
Suppose $a > 0$ and $a^2 + b^2 < \nu_0^2$.  Then for all $\gamma$ sufficiently small, there exist unique $\lambda = \hat \lambda(\gamma) < 0$ and $\mu = \hat \mu(\gamma) \in \bbR$ satisfying \eqref{eq:ext_problem}. 

Moreover, if $(\la_0, \mu_0)$ is the unique solution to the inextensible problem \eqref{eq:inext} with $\gamma_0 = 1$, then
\begin{align}
\hat \lambda(\gamma) = \gamma \lambda_0 + O(\gamma^2), \quad 
\hat \mu(\gamma) = \gamma \mu_0 + O(\gamma^2). \label{eq:lamlam0}
\end{align} 
In particular, for all $\gamma$ sufficiently small, the catenary is completely extensible, $\hat \nu(s) \in (\nu_0,1)$ for all $s \in [0,1]$.
\end{prop}

\begin{lem}\label{l:lamu_to0}
Suppose $a > 0$ and $a^2 + b^2 < \nu_0^2$. Let $\gamma_n > 0$ and $(\la_n,\mu_n)$ satisfy \eqref{eq:ext_problem} with $\gamma_n \rar 0$ as $n \rar \infty$. Then 
\begin{align}
\lim_{n \rar \infty} \sqrt{\la_n^2 + \mu_n^2} = 0. 
\end{align}
\end{lem} 

\begin{proof}
Suppose not. Then for all $n \in \bbN$, there exist $\gamma_n > 0$,  $\lambda_n < 0$ and $\mu_n \in \bbR$ satisfying \eqref{eq:ext_problem} such that $\gamma_n \rar 0$ and 
\begin{align}
\sqrt{\la_n^2 + \mu_n^2} \rar \delta_0 \in (0,\infty]. \label{eq:bad_limit}
\end{align}
We consider the two cases $\delta_0 = \infty$ and $\delta_0 \in (0,\infty)$ separately. 

If $\delta_0 = \infty$, then for all $n$ sufficiently large, for all $s \in [0,1]$
\begin{align}
-\sqrt{\lambda_n^2 + (\mu_n+g\gamma_n s)^2} < N_0. \label{eq:inext_ass}
\end{align} Then \eqref{eq:ext_problem} implies 
\begin{align*}
\frac{a}{\nu_0} \bs i + \frac{b}{\nu_0} \bs k
&= \int_0^1 -\frac{1}{\sqrt{\la_n^2 + (\mu_n + g\gamma_n s)^2}}
\bigl (\la_n \bs i + (\mu_n + g \gamma_n s) \bs k \bigr ) ds,
\end{align*}
By \eqref{eq:lambda_equation} and \eqref{eq:mu_equation}, $\la_n$ is determined by the relation 
\begin{align}
\frac{\sqrt{\nu_0^2 - b^2}}{a} = \frac{2 \nu_0 |\la_n|}{a g \gamma_n} \sinh \frac{a g \gamma_n}{2 \nu_0 |\la_n|}, \label{eq:lambda_eq}
\end{align}
and $\mu_n$ is determined by 
\begin{align}
\mu_n = \la_n \sinh \Bigl (
\frac{a g \gamma_n}{2 \nu_0 |\la_n|} + \tanh^{-1} \frac{b}{\nu_0}
\Bigr ). \label{eq:mu_eq}
\end{align}
Let $z \in (0,\infty)$ be the unique solution to $\frac{\sqrt{\nu_0^2 - b^2}}{a} = \frac{\sinh z}{z}$, which exists since $\sqrt{\nu_0^2 - b^2} > a$. Then by \eqref{eq:lambda_eq}  
\begin{align}
\frac{ag \gamma_n}{2 \nu_0 |\la_n|} = z
\end{align}
whence $\la_n = O(\gamma_n)$. By \eqref{eq:mu_eq} it follows that 
\begin{align}
\mu_n = \la_n\sinh \Bigl (z + \tanh^{-1}\frac{b}{\nu_0} \Bigr )
\end{align}
whence $\mu_n = O(\gamma_n)$. Thus, 
\begin{align*}
\sqrt{\lambda_n^2 + \mu_n^2} \rar 0 \mbox{ as } n \rar \infty
\end{align*}
 which contradicts \eqref{eq:inext_ass}.
 
We now consider the case 
\begin{align}
\sqrt{\lambda_n^2 + \mu_n^2} \rar \delta_0 \in (0,\infty). \label{eq:slightlyext_ass}
\end{align}
Passing to a subsequence and relabeling if necessary, we can assume that there exists $(\la_*, \mu_*)$ such that 
\begin{align}
\lim_{n \rar \infty} (\la_n,\mu_n) = (\la_*,\mu_*). 
\end{align}
By \eqref{eq:slightlyext_ass}, we conclude
\begin{align}
\sqrt{\la_*^2 + \mu_*^2} = \delta_0 > 0, 
\end{align}
and therefore, for all $s \in [0,1]$
\begin{align*}
\lim_{n \rar \infty} -\sqrt{\la_n^2 + (\mu_n + g \gamma_n s)^2} = 
-\sqrt{\la_*^2 + \mu_*^2} = -\delta_0,  
\end{align*}
By the continuity of $\hat \nu$, the dominated convergence theorem and \eqref{eq:ext_problem} we conclude 
\begin{align*}
a \bs i + b \bs k
&= \lim_{n \rar \infty} \int_0^1 \frac{\hat \nu(-\sqrt{\la_n^2 + (\mu_n + g \gamma_n s)^2})}{-\sqrt{\la_n^2 + (\mu_n + g\gamma_n s)^2}}
\left [\la_n \bs i + (\mu_n + g\gamma_n s) \bs k \right ] ds \\
&= - \hat \nu(-\delta_0) \left [
\frac{\la_*}{\sqrt{\la_*^2 + \mu_*^2}} \bs i + 
\frac{\mu_*}{\sqrt{\la_*^2 + \mu_*^2}} \bs k
\right ]
\end{align*}
and thus 
\begin{align}
a^2 + b^2 = \hat \nu(-\delta_0)^2 \geq \nu_0^2, 
\end{align}
a contradiction to the assumption $a^2 + b^2 < \nu_0^2$. Thus, \eqref{eq:bad_limit} cannot hold which proves the lemma. 
\end{proof}

Via rescaling the variables $(\la,\mu)$ and the implicit function theorem, we have the following result from \cite{Wolfe97}. 

\begin{lem}\label{l:slightly_ext}
	Suppose $a > 0$ and $a^2 + b^2 < 1$. Let $(\la_0, \mu_0)$ be the unique solution to the inextensible problem \eqref{eq:inext} with $\gamma_0 = 1$.  There exists $\eps > 0$ and $\zeta > 0$ such that if $\gamma \in (0,\epsilon)$, then there exists a unique pair 
	\begin{align*}
	(\hat \la(\gamma), \hat \mu(\gamma)) \in \left \{ \left |\frac{\lambda}{\gamma} - \lambda_0 \right |^2 + \left | \frac{\mu}{\gamma} - \mu_0 \right |^2 < \zeta^2  \right \},
	\end{align*}
	satisfying \eqref{eq:ext_problem}. Moreover, 
	\begin{align}
	\hat \lambda(\gamma) = \gamma \lambda_0 + O(\gamma^2), \quad  
	\hat \mu(\gamma) = \gamma \mu_0 + O(\gamma^2). 
	\end{align}
\end{lem}

\begin{proof}[Proof of Proposition \ref{p:uniqueness}]
Let $(\lambda, \mu)$ satisfy \eqref{eq:ext_problem} (we now drop the dependence in $\gamma$). We claim that 
\begin{align}
\lim_{\gamma \rar 0} \frac{\lambda}{\gamma} = \lambda_0, \quad
\lim_{\gamma \rar 0} \frac{\mu}{\gamma} = \mu_0. \label{eq:lammu_convergence}
\end{align}
Then \eqref{eq:lammu_convergence} and Lemma \ref{l:slightly_ext} immediately imply the conclusions Proposition \ref{p:uniqueness}. 

To prove \eqref{eq:lammu_convergence}, let $\gamma_n > 0$ and $(\lambda_n, \mu_n)$ satisfy \eqref{eq:ext_problem}, and suppose that $\gamma_n \rar 0$. We wish to prove that 
\begin{align}
\frac{\la_n}{\gamma_n} \rar \la_0, \quad \frac{\mu_n}{\gamma_n} \rar \mu_0. \label{eq:lanmun_convergence}
\end{align} 
Define 
\begin{align*}
\delta_n(s) &:= \sqrt{\la_n^2 + (\mu_n + g\gamma_n s)^2}, \\
a_n &:= a + \int_0^1 (\hat \nu(-\delta_n(s)) - 1) \frac{\la_n}{\delta_n(s)} ds, \\
b_n &:= b + \int_0^1 (\hat \nu(-\delta_n(s)) - 1) \frac{\mu_n + g \gamma_n s}{\delta_n(s)} ds, \\
\la_{0,n} &:= \frac{\la_n}{\gamma_n}, \quad \mu_{0,n} := \frac{\mu_n}{\gamma_n}, \\
\delta_{0,n}(s) &:= \sqrt{\la_{0,n}^2 + (\mu_{0,n} + g s)^2}. 
\end{align*}
By Lemma \ref{l:lamu_to0}, $(\la_n, \mu_n) \rar (0,0)$ which implies $\delta_n(s) \rar 0$ uniformly on $[0,1]$. Since $\hat \nu(\cdot)$ is continuous, we conclude 
\begin{align}
\lim_{n \rar \infty} a_n = a, \quad \lim_{n \rar \infty} b_n = b. \label{eq:anbn}
\end{align}
From the above definitions and \eqref{eq:ext_problem}, we have 
\begin{align}
a_n \bs i + b_n \bs k = \int_0^1 -\frac{1}{\delta_{0,n}(s)} \bigl (
\la_{0,n} \bs i + (\mu_{0,n} + g s) \bs k
\bigr ),
\end{align}
and thus, by \eqref{eq:lambda_equation} and \eqref{eq:mu_equation}, $(\la_{0,n}, \mu_{0,n})$ are uniquely determined by the relations 
\begin{align}
\frac{\sqrt{1 - b_n^2}}{a_n} &= \frac{2|\lambda_{0,n}|}{a_n g} \sinh \frac{a_n g}{2 |\la_{0,n}|}, \label{eq:lan_equation}\\
\mu_{0,n} &= \lambda_{0,n}\sinh \left (
\frac{a g}{2 |\lambda_{0,n}|} + \tanh^{-1} b_n
\right ). \label{eq:mun_equation}
\end{align}
Since $a_n \rar a$ and $b_n \rar b$, \eqref{eq:lan_equation} implies $\la_{0,n} \rar \la_0$ i.e. $\frac{\la_n}{\gamma_n} \rar \la_0$. By \eqref{eq:mun_equation}, it then follows that $\mu_{0,n} \rar \mu_0$ i.e. $\frac{\mu_n}{\gamma_n} \rar \mu_0$. This concludes the proof of \eqref{eq:lanmun_convergence} and Proposition \ref{p:uniqueness}.

\end{proof} 

\section{Conclusion} 

This paper considers stationary strings suspended between two supports under the force of gravity (catenaries) satisfying a new stretch-limiting constitutive relation. We explicitly classify the positions of the supports leading to tensile states containing fully stretched, inextensible segments in two cases: the degenerate case when the string is vertical and straight, and the nondegenerate case when the supports are at the same height. We then turn to the question of multiplicity of compressive states in general and prove uniqueness of compressive states when the distance between supports is less than the minimal length of the string. This work should be viewed as an early first step in exploring the mathematical properties of implicit constitutive relations within the realm of one-dimensional elastic bodies including strings and, more generally, rods. In particular, we have considered the simplest stationary setting of stretch-limited elastic strings, leaving the study of dynamical motion for future work. 

\bibliographystyle{plain}
\bibliography{researchbibmech}
\bigskip

\centerline{\scshape Casey Rodriguez}
\smallskip
{\footnotesize
 \centerline{Department of Mathematics, Massachusetts Institute of Technology}
\centerline{77 Massachusetts Ave, 2-246B, Cambridge, MA 02139, U.S.A.}
\centerline{\email{caseyrod@mit.edu}}
}

\end{document}